\documentclass[a4paper, conference]{IEEEtran}
\usepackage[left=1.62cm,right=1.62cm,top=1.9cm,bottom=4.5cm]{geometry}
\usepackage{amsthm}
\usepackage{amsfonts}
\usepackage{amssymb}
\usepackage{mathrsfs}
\usepackage{mathtools}
\usepackage{float}
\usepackage[pdftex]{graphicx}
\usepackage{amsmath}
\usepackage{color}
\usepackage{multirow,multicol}
\usepackage{graphicx}
\usepackage{times}
\usepackage{textcomp}
\usepackage{verbatim}
\usepackage[table]{xcolor}% http://ctan.org/pkg/xcolor
\usepackage{balance}
\usepackage{lipsum}
\usepackage[inline]{enumitem}
\usepackage[caption=false,font=footnotesize]{subfig}
\usepackage{cite}
\usepackage{algpseudocode,algorithm}
\usepackage{hyperref}
\usepackage{setspace,epstopdf}
\usepackage[table]{xcolor}% http://ctan.org/pkg/xcolor
\definecolor{color1}{RGB}{199,209,232}
\definecolor{color2}{RGB}{230,231,233}
\DeclareMathOperator*{\argmax}{argmax} % thin space, limits underneath in displays
\DeclareMathOperator*{\minimize}{minimize} % thin space, limits underneath in 
\DeclareMathOperator*{\subjectto}{subject\hspace{3pt} to:\hspace{3pt}} % thin space, 
\newtheorem{theorem}{Theorem}

\begin{document}
	
	\title{Near-field Hybrid Beamforming for Terahertz-band Integrated Sensing and Communications}
	
	\author{
		\IEEEauthorblockA{Ahmet M. Elbir$^{\dag,+}$,  Abdulkadir Celik$^{+}$ and Ahmed M. Eltawil$^{+}$}
		\IEEEauthorblockA{
			${\dag}$University of Luxembourg, Luxembourg \\
			$+$King Abdullah University of Science and Technology, Saudi Arabia}
		\IEEEauthorblockA{E-mail:\texttt{ ahmetmelbir@ieee.org, abdulkadir.celik@kaust.edu.sa, ahmed.eltawil@kaust.edu.sa} }
	}
	
	\maketitle

	% Reduce spacing above and below equations
	%	\setlength{\abovedisplayskip}{3pt}
	%	\setlength{\belowdisplayskip}{3pt}

	%	\red{AME: I'm revising the paper offline, will upload the updated version.}
	
	%	{\color{red} Why mmWave/THz? 
	%		
	%		If we do THz-only SPIM-ISAC, we add wideband and investigate beam-split, which introduce more novelty. But, in Thz, the number of paths is small, e.g., 5~\cite{ummimoTareq}, so employing SPIM may be questioned by the reviewers. Because SPIM is meaningful if the environment is rich in terms of number of paths, which makes mmWave more applicable. 
	%		
	%		If we do mmWave-only SPIM-ISAC, the channel would be path-rich and no question about the usage of SPIM. We can add wideband and investigate beam-squint. But the problem would not be much new, and beam-squint is not much severe in mmWave, e.g., $f_c=60$ GHz, $B = 2$ GHz, so comparing the results of beam-squint-free and beam-squint-affected would be very close. 
	%		
	%		Thus, I used mmWave/THz  since our approach is applicable for both. I revised the paper accordingly, touching both mmWave and THz. The channel model in (\ref{channel}) is described for both scenarios. Then, I added the SE versus bandwidth graph in Fig.~\ref{fig_BS} only for THz scenario. 
	%		
	%		Maybe we don't say mmWave/THz in the title, but we explain it in the abstract and the intro.
	%		
	%		
	%		
	%	} 
	%%	
	
	\begin{abstract}
		Terahertz (THz) band communications and integrated sensing and communications (ISAC) are two main facets of the sixth generation wireless networks. In order to compensate the severe attenuation, the THz wireless systems employ large arrays, wherein the near-field beam-squint severely degrades the beamforming accuracy. Contrary to prior works that examine only either narrowband ISAC beamforming or far-field models, we introduce an alternating optimization technique for hybrid beamforming design in near-field THz-ISAC scenario. We also propose an efficient approach to compensate near-field beam-squint via baseband beamformers. Via numerical simulations, we show that the proposed approach achieves satisfactory spectral efficiency performance while accurately estimating the near-field beamformers and mitigating the beam-squint without additional hardware components. 
		
	\end{abstract}


	\section{Introduction}
	\label{sec:Introduciton}
	\IEEEPARstart{I}{ntegrated} sensing and communications (ISAC) has emerged as one of the pivotal technologies of future sixth generation (6G) wireless networks, enabling synergistic access to the scarce radio spectrum on an integrated hardware platform ~\cite{mishra2019toward,elbir2022Aug_THz_ISAC}. In particular, as the allocation of the spectrum beyond 100 GHz is underway, specifically in the terahertz (THz) band, ISAC is currently witnessing frantic research endeavors to simultaneously achieve high-resolution sensing and ultrahigh-speed communications system architecture at the THz frequencies~\cite{elbir2021JointRadarComm,elbir2022Aug_THz_ISAC}.

	Signal processing  at THz-band confronts  multiple impediments, such as severe path loss, limited transmission distance, and \textit{beam-squint}. To surmount these challenges at reduced hardware costs, hybrid analog and digital beamforming architectures are employed in a massive multiple-input multiple-output (MIMO) array configuration~\cite{heath2016overview,elbir2022Nov_Beamforming_SPM}. For higher spectral efficiency (SE) and lower complexity, massive MIMO systems employ wideband signal processing, wherein subcarrier-dependent (SD) baseband and subcarrier-independent (SI) analog beamformers are adopted. In particular, the weights of the analog beamformers are subject to a single (sub-)carrier frequency~\cite{alkhateeb2016frequencySelective}. Therefore, the beam generated across the subcarriers points towards disparate directions, engendering {beam-squint} phenomenon~\cite{elbir_THZ_CE_ArrayPerturbation_Elbir2022Aug,beamSquint_FeiFei_Wang2019Oct}.  Compared to millimeter-wave (mm-Wave) frequencies, beam-squint's ramifications are more acute in THz massive MIMO  because of wider system bandwidths in the latter~\cite{beamSquint_FeiFei_Wang2019Oct,elbir_BSA_OMP_THZ_CE_Elbir2023Feb}. %Thus, beam-squint must be addressed for reliable system performance. 	
	As such, addressing beam-squint is imperative for ensuring reliable system performance. Existing techniques to compensate for the impact of beam-squint mostly employ additional hardware components, e.g., time-delayer (TD) networks~\cite{trueTimeDelayBeamSquint,beamSquint_FeiFei_Wang2019Oct} and SD phase shifter networks~\cite{beamSquintAwareHB_SD_You2022Aug} to virtually realize SD analog beamformers. % via SD signal processing. 
	However, these approaches are inefficient in terms of cost and power~\cite{elbir2022Aug_THz_ISAC}. It merits notingthat beam-squint compensation does not necessitate additional hardware components for estimation of the communications channel and radar target direction-of-arrival, which can be handled in the digital domain, wherein the generation of SD analog beamformers is possible. Nevertheless, supplementary (analog) hardware is required for hybrid (analog/digital) beamformer design~\cite{elbir2022Nov_Beamforming_SPM,elbir2023Mar_ISAC_SPIM}. 
	% problem since the process is not completely in the digital domain. 
	%	The aforementioned THz works~\cite{dovelos_THz_CE_channelEstThz2,elbir2022Jul_THz_CE_FL,elbir_THZ_CE_ArrayPerturbation_Elbir2022Aug,spatialWidebandWang2018May,elbir2022_thz_beamforming_Unified_Elbir2022Sep,thz_channelEst_beamsplitPatternDetection_L_Dai} as well as  the conventional wireless systems operating at sub-$6$ GHz and mm-Wave bands~\cite{limitedFeedback_Alkhateeb2015Jul,channelEstLargeArrays2,channelEstLargeArrays,channelEstimation1}  mostly incorporate far-field plane-wave model whereas the transmission range is shorter in THz-band such that the users are usually in the near-field region~\cite{ummimoTareqOverview}. 
	
	Beside beam-squint, another formidable challenge in THz-band signal processing is short-transmission distance, which may cause the signal wavefront at the receive to become spherical in near-field (see, e.g., Fig.~\ref{fig_NF_ISAC}). In particular, the plane wavefront is spherical in the near-field when the transmission range is shorter than the Fraunhofer distance~\cite{nf_primer_Bjornson2021Oct}. As a result, the beamforming algorithms must accommodate the near-field model, which depends on both direction and range information  for accurate signal processing~\cite{elbir2022Aug_THz_ISAC}. Among the works investigating the near-field signal  model,  \cite{nf_OMP_Dai_Wei2021Nov,nf_mmwave_CE_noBeamSplit_Cui2022Jan,nf_NB2_Zhang2022Nov} consider the near-field scenario,
	%while the effect of beam-squint is ignored and only mm-Wave scenario is investigated. 
	while neglecting the effect of beam-squint and focusing solely on mm-Wave scenarios. On the other hand,  several methods have been proposed to compensate the far-field beam-squint for both THz channel estimation~\cite{elbir2022Jul_THz_CE_FL,dovelos_THz_CE_channelEstThz2} and beamforming~\cite{elbir2021JointRadarComm,elbir2022_thz_beamforming_Unified_Elbir2022Sep} applications. Furthermore, near-field THz channel estimation is explored in~\cite{elbir2023Feb_NF_THZ_CE,elbir2023Feb_NF_THZ_CE_ICASSP_NBAOMP}, wherein an orthogonal matching pursuit (OMP)-based approach is proposed. The near-field ISAC scenario is investigated in~\cite{nf_ISAC_1Wang2023Feb}, albeit exclusively for narrowband systems that do not account for the impact of beam-squint. Specifically, \cite{nf_ISAC_1Wang2023Feb} considers a near-field multiple signal classification (MUSIC) algorithm to estimate the direction and ranges of radar targets and communication users. Nevertheless, near-field ISAC hybrid beamforming in the presence of beam-squint remains relatively unexamined.
	
	In this paper, near-field hybrid beamforming approach is proposed for the THz-ISAC scenario. We first introduce the system model for both communications and sensing signal acquisition. Subsequently, the near-field array model and near-field beam-squint are introduced. In order to design the hybrid beamformers, an alternating algorithm is devised. Initially, a dictionary of near-field steering vectors is employed to estimate the analog beamformer. Then, the baseband beamformer and the joint radar-communications (JRC) beamformers are estimated. Finally, we introduce  an efficient approach to compensate beam-squint in the baseband rather than designing SD analog beamformers~\cite{beamSquintAwareHB_SD_You2022Aug} or TD networks~\cite{trueTimeDelayBeamSquint,beamSquint_FeiFei_Wang2019Oct}, which are hardware-inefficient. Specifically, we design a beam-squint-aware (BSA) baseband beamformer by matching the SI hybrid beamformer to the SD one. Therefore, the effect of beam-squint is conveyed from analog domain to the baseband. 
	
	%	first an OMP-based approach is devised to construct \textit{virtual} SD analog beamformers. Then, a single SI analog beamformer is obtained via matching this virtual SD hybrid beamformer to the unconstrained joint radar-communications (JRC) beamformer, for which we propose a beam-squint-aware (BSA) baseband beamformer, which is SD. Therefore, the beam-squint is handled by the baseband beamformer, 
	%%-----------------------------------------------------
	\begin{figure}
		\centering
		{\includegraphics[draft=false,width=\columnwidth]{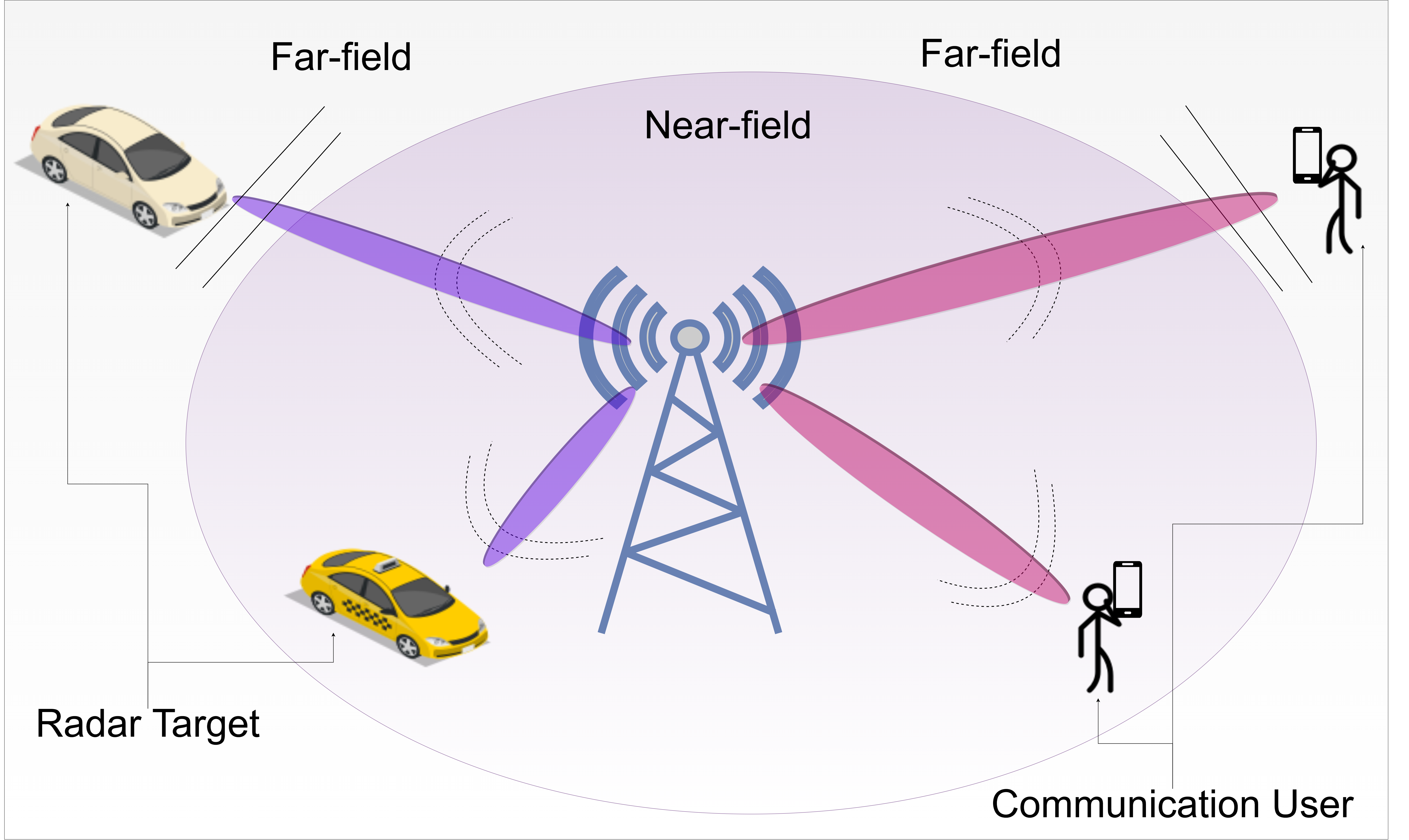} } 
		% 	\subfloat[]{\includegraphics[draft=false,width=.33\textwidth]{BS_frequency.eps} } 
		\caption{Near-field ISAC scenario, wherein received signal wavefront for the communication user/targets in far-field (near-field) is plane-wave (spherical-wave).   	}
		%			\vspace*{-5mm}
		\label{fig_NF_ISAC}
	\end{figure}
	%%-----------------------------------------------------
	%	\subsection{Notation}
	%	In the remainder of the paper, we first present the signal and system model for the proposed approach in Sec.~\ref{sec:Model}. Then, we introduce the proposed SPIM-ISAC approach in Sec.~\ref{sec:SPIMISAC} which involves parameter estimation for radar (Sec.~\ref{sec:RadarParEst}) and communications (Sec.~\ref{sec:CommParEst}) as well as hybrid beamformer design (Sec.~\ref{sec:HybBF}). The simulation results are presented in Sec.~\ref{sec:Sim}, and the paper is finalized with conclusions in Sec.~\ref{sec:Conc}. 
	
	\textit{Notation:}  Throughout the paper,  $(\cdot)^\textsf{T}$ and $(\cdot)^{\textsf{H}}$ denote the transpose and conjugate transpose operations, respectively. For a matrix $\mathbf{A}$ and vector $\mathbf{a}$; $[\mathbf{A}]_{i,j}$, $[\mathbf{A}]_k$  and $[\mathbf{a}]_l$ correspond to the $(i,j)$-th entry, $k$-th column and $l$-th entry, respectively.  An $N\times N$ identity matrix is represented by $\mathbf{I}_{N} $. The pulse-shaping function is represented by $\mathrm{sinc}(t) = \frac{\pi t}{t}$. We denote $|| \cdot||_2$ and $|| \cdot||_\mathcal{F}$ as the  $l_2$-norm and Frobenious norm, respectively.

	%	\vspace{-12pt}
	\section{System Model \& Problem Formulation}
	\label{sec:Model}
	Consider a wideband transmitter design problem in an ISAC scenario with a communication user and $K$ radar targets located in the near-field of the base station (BS) as illustrated in Fig.~\ref{fig_NF_ISAC}. The dual-function BS jointly communicate with the communication user and sense the radar targets via probing signals with $N_\mathrm{T}$ antennas over $M$ subcarriers. The user has $N_\mathrm{R}$ antennas, for which $N_\mathrm{S}$ data symbols $\mathbf{s}[m] = [s_1[m],\cdots,s_{N_\mathrm{S}}[m]]^\textsf{T}\in \mathbb{C}^{N_\mathrm{S}}$ are transmitted, where $\mathbb{E}\{\mathbf{s}[m]\mathbf{s}^\textsf{H}[m]\}=1/{N_\mathrm{S}}\mathbf{I}_{N_\mathrm{S}}$. 
	
	\subsection{Communications Model}
	The BS aims to transmit the data symbol vector $\mathbf{s}[m]\in \mathbb{C}^{N_\mathrm{S}}$ toward the communications user. Thus, the BS first applies the SD baseband beamformer $\mathbf{F}_\mathrm{BB}[m]\in\mathbb{C}^{N_\mathrm{RF}\times N_\mathrm{S}}$. Then,  $M$-point inverse fast Fourier transform (IFFT) is applied to convert the signal to time-domain, and the cyclic prefix (CP) is added. Finally, the SI analog beamformer ${\mathbf{F}}_\mathrm{RF}\in \mathbb{C}^{N_\mathrm{T}\times N_\mathrm{RF}}$ is applied, and the $N_\mathrm{T}\times 1$ transmit signal  becomes
	\begin{align}
	\mathbf{x}[m] = \mathbf{F}_\mathrm{RF}\mathbf{F}_\mathrm{BB}[m]\mathbf{s}[m],
	\end{align}
	where the analog beamformer $\mathbf{F}_\mathrm{RF}$ has constant-modulus constraint, i.e., $|[\mathbf{F}_\mathrm{RF}]_{i,j}| = 1/\sqrt{N_\mathrm{T}}$ for $i = 1,\cdots, N_\mathrm{T}$, $j = 1,\cdots,N_\mathrm{RF}$. Furthermore, we have $\sum_{m=1}^{M}\| \mathbf{F}_\mathrm{RF}\mathbf{F}_\mathrm{BB}[m]\|_\mathcal{F}^2 = MN_\mathrm{S}$ to account for the total power constraint.
	%	= \mathbf{F}_\mathrm{RF} \mathbf{D}^{(i)} Here, $\mathbf{F}_\mathrm{RF}\in \mathbb{C}^{N_\mathrm{T}\times \bar{P}}$ is the \textit{complete} analog beamformer matrix comprised of all possible spatial paths. 

	%\vspace{-12pt}
	\subsubsection{THz Channel Model} In this study,  we employ Saleh-Valenzuela (S-V) multipath channel model, which is the superposition of received non-LoS (NLoS) paths to model the THz channel~\cite{ummimoTareqOverview,ummimoHBThzSVModel}.	Compared to the mmWave channel, the THz channel involves limited reflected paths and negligible scattering~\cite{ummimoTareqOverview,thz_mmWave_path_Comparison_Yan2020Jun}. For massive MIMO systems, approximately $5$ paths survive at $0.3$ THz compared to approximately $8$ paths at $60$ GHz~\cite{thz_mmWave_path_Comparison_Yan2020Jun}. Especially for outdoor applications, multipath channel models are widely used to represent the THz channel for a more general scenario~\cite{ummimoTareqOverview,thz_mmWave_path_Comparison_Yan2020Jun}.  Hence, in this work, we consider a general scenario, wherein the delay-$\bar{d}$ $N_\mathrm{R}\times N_\mathrm{T}$ MIMO communications channel involving $L$ NLoS paths is given in discrete-time domain as
	\begin{align}
	\label{channelTimeDomain}
	\tilde{\mathbf{H}}(\bar{d}) = \sum_{l = 1}^{L} {\gamma}_l \mathrm{sinc}(\bar{d} - B\tau_l) \mathbf{a}_\mathrm{R}(\theta_l,  \rho_{l}) \mathbf{a}_\mathrm{T}^\textsf{H}(\phi_l,r_l), 
	\end{align}
	where ${\gamma}_l\in \mathbb{C}$ denotes the channel path gain, $B$ represents the system bandwidth and $\tau_l$ is the time delay of the $l$-th path.  $\theta_l$ ($\rho_l$) and $\phi_l$ ($r_l$) denote the physical DoA and direction-of-departure (DoD) angles (ranges) of the scattering paths between the user and the BS, respectively, where  $\theta_l = \sin \tilde{\theta}_l$, $\phi_l = \sin \tilde{\phi}_l$ and $\tilde{\theta}_l,\tilde{\phi}_l \in [-\frac{\pi}{2},\frac{\pi}{2}]$. Then, the corresponding receive and transmit steering vectors are defined as $\mathbf{a}_\mathrm{R}(\theta_l, \rho_l)\in \mathbb{C}^{N_\mathrm{R}}$ and $\mathbf{a}_\mathrm{T}(\phi_l,r_l)\in \mathbb{C}^{N_\mathrm{T}}$, respectively. Performing $M$-point FFT of the delay-$\bar{d}$ channel given in (\ref{channelTimeDomain}) yields 
	\begin{align}
	\mathbf{H}[m] = \sum_{\bar{d}=1}^{\bar{D}-1} \tilde{\mathbf{H}}(\bar{d}) e^{- \mathrm{j}\frac{2\pi m}{M} \bar{d} },
	\end{align} where $\bar{D}\leq M$ is the CP length.
	Then, the $N_\mathrm{R}\times N_\mathrm{T}$ channel matrix in frequency domain is represented by
	\begin{align}
	\label{channelFrequencyDomain}
	\mathbf{H}[m] \hspace{-3pt}=\hspace{-3pt} \sum_{l = 1}^{L} {\gamma}_l \mathbf{a}_\mathrm{R}(\bar{\theta}_{l,m},\bar{\rho}_{l,m}) \mathbf{a}_\mathrm{T}^\textsf{H}(\bar{\phi}_{l,m},\bar{r}_{l,m})e^{-\mathrm{j}2\pi \tau_l f_m},
	\end{align}
	where $\bar{\theta}_{l,m}$ ($\bar{\rho}_{l,m}$) and $\bar{\phi}_{l,m}$ $(\bar{r}_{l,m})$ denote the spatial directions (ranges), which are SD and they are deviated from the physical directions $\theta_l$, $\phi_l$ ($\rho_l,r_l$) in the beamspace due to beam-squint~\cite{beamSquint_FeiFei_Wang2019Oct,elbir2022Aug_THz_ISAC}. On the other hand, the beam-squint-free channel matrix is
	\begin{align}
	\overline{\mathbf{H}}[m] = \sum_{l = 1}^{L} {\gamma}_l \mathbf{a}_\mathrm{R}(\phi_{l},\rho_l) \mathbf{a}_\mathrm{T}^\textsf{H}(\theta_{l},r_l)e^{-\mathrm{j}2\pi \tau_l f_m}.
	\end{align}Then, the $N_\mathrm{R}\times 1$ received signal at the communications user is
	\begin{align}
	\mathbf{y}[m] = \mathbf{H}[m]\mathbf{F}_\mathrm{RF}\mathbf{F}_\mathrm{BB}[m]\mathbf{s}[m] + \mathbf{n}[m],
	\end{align}
	where $\mathbf{n}[m]\sim \mathcal{CN}(\mathbf{0},\sigma_n^2\mathbf{I}_{N_\mathrm{R}})\in\mathbb{C}^{N_\mathrm{R}}$ represents the temporarily and spatially  additive white Gaussian noise vector.

	%\vspace{-12pt}
	\subsubsection{Beam-Squint Effect}
	\label{sec:beamSplit} 
	%	In conventional wideband systems, the operating bandwidth is relatively small and the subcarrier frequencies are close to each other, i.e., $f_{m_1} \approx f_{m_2}$. Therefore, a single wavelength assumption, i.e., $\lambda_1 = \cdots, \lambda_M = \frac{c_0}{f_c}$, is made across the subcarriers, where $c_0$ and $f_c$ are the speed of light and carrier frequency, respectively. Thus, a single analog beamformer is usually used for all subcarriers in wideband mmWave systems~\cite{heath2016overview,alkhateeb2016frequencySelective,beamSquintWang2019Nov,trueTimeDelayBeamSquint,beamSquint_FeiFei_Wang2019Oct}. 

	In wideband transmission, %a single wavelength assumption, i.e., $\lambda_1 = \cdots \lambda_M = \frac{c_0}{f_c}$, is usually made across the subcarriers, where $c_0$ and $f_c$ are the speed of light and carrier frequency, respectively. 
	the prevalent assumption is the employment of a monochromatic wavelength across all subcarriers, delineated as $\lambda_1 = \cdots \lambda_M = \frac{c_0}{f_c}$, wherein $c_0$ signifies the speed of light and $f_c$ represents the carrier frequency. However, 
	%due to employing a single analog beamformer, the single wavelength assumption does not hold, and the generated beams by employing a single analog beamformer are squinted and point to different directions/ranges in the spatial domain
	the utilization of a singular analog beamformer renders this monochromatic wavelength assumption inapplicable, culminating in the formation of squinted beams that orient toward disparate spatial directions and ranges~\cite{beamSquint_FeiFei_Wang2019Oct,elbir2022Aug_THz_ISAC}. %Suppose that similar beamforming architecture 
	Presuming an analogous beamforming architecture is employed at the user end(i.e., SI analog beamformer with SD digital beamformers), 
	%is employed at the user.Furthermore, 
	the high-frequency operation at THz implies the the presence of close-proximity users in the near-field region, where planar wave propagation is not valid. At ranges shorter than the Fraunhofer distance $d_F = \frac{2 D^2}{\lambda}$, where $D$ is the array aperture and $\lambda = \frac{c_0}{f_c}$ is the wavelength, the near-field wavefront exhibits spherical nature \cite{nf_primer_Bjornson2021Oct,elbir_THZ_CE_ArrayPerturbation_Elbir2022Aug}.  For a uniform linear array (ULA), the array aperture is $D = (N-1)d$, where $d = \frac{\lambda}{2}$ is the element spacing. In the THz spectrum, it is imperative to employ a near-field signal model because $r_{l} <d_F$. For instance, when $f_c = 300$ GHz and $N=256$, the Fraunhofer distance is $d_F = 32.76$ m.
	
	Taking into account the spherical-wave model~\cite{nf_primer_Bjornson2021Oct,nf_Fresnel_Cui2022Nov,elbir2023Feb_NF_THZ_CE}, we define the near-field steering vector $\mathbf{a}_\mathrm{T}(\phi_{l},r_{l})\in\mathbb{C}^{N_\mathrm{T}}$ corresponding to the physical DoA  $\phi_{l}$ and range $r_{l}$ as 
	\begin{align}
	\label{steeringVec1}
	\mathbf{a}_\mathrm{T}(\phi_{l},r_{l}) = \frac{1}{\sqrt{N_\mathrm{T}}} [e^{- \mathrm{j}2\pi \frac{d}{\lambda}r_{l}^{(1)} },\cdots,e^{- \mathrm{j}2\pi \frac{d}{\lambda}r_{l}^{(N_\mathrm{T})} }]^\textsf{T},
	\end{align}
	where  $r_{l}^{(n)}$ is the distance between the $l$-th path scatterer and the $n$-th antenna as
	\begin{align}
	r_{l}^{(n)} = \left(r_{l}^2  + 2(n-1)^2 d^2 - 2 r_{l}(n-1) d \phi_{l}   \right)^{\frac{1}{2}}. \label{eq:rkln}
	\end{align}
	Following the Fresnel approximation~\cite{nf_Fresnel_Cui2022Nov,elbir2023Feb_NF_THZ_CE}, \eqref{eq:rkln} becomes
	\begin{align}
	\label{r_approx}
	r_{l}^{(n)} \approx r_{l}  - (n-1) d \phi_{l}  + (n-1)^2 d^2 \zeta_{l}  ,
	\end{align}	 
	where $\zeta_{l} = \frac{1- \phi_{l}^2}{2 r_{l}}$. Rewrite (\ref{steeringVec1}) as
	\begin{align}
	\label{steeringVectorPhy}
	\mathbf{a}_\mathrm{T}(\phi_{l},r_{l}) \approx e^{- \mathrm{j}2\pi \frac{f_c}{c_0}r_{l}} \tilde{\mathbf{a}}_\mathrm{T}(\phi_{l},r_{l}),
	\end{align} where the $n$-th element of $\tilde{\mathbf{a}}_\mathrm{T}(\phi_{l},r_{l})\in \mathbb{C}^{N_\mathrm{T}}$ is 
	\begin{align}
	\label{steeringVectorPhy2}
	[\tilde{\mathbf{a}}_\mathrm{T}(\phi_{l},r_{l})]_n = e^{\mathrm{j} 2\pi \frac{f_c}{c_0}\left( (n-1)d\phi_{l}  - (n-1)^2 d^2 \zeta_{l}\right) }.
	\end{align}
	The steering vector in (\ref{steeringVectorPhy}) corresponds to the physical location $(\phi_{l},r_{l})$. This deviates to the spatial location $(\bar{\phi}_{m,l},\bar{r}_{m,l})$ in the beamspace because of the absence of SD analog beamformers. Then, the $n$-th entry of the deviated steering vector in (\ref{steeringVectorPhy2}) for the spatial location is 
	\begin{align}
	\label{steeringVectorSpa}
	&[\tilde{\mathbf{a}}_\mathrm{T}(\bar{\phi}_{m,l},\bar{r}_{m,l})]_n \hspace{-3pt}= \hspace{-2pt}e^{\mathrm{j} 2\pi \frac{f_m}{c_0}\left( (n-1)d\bar{\phi}_{m,l}  - (n-1)^2 d^2 \bar{\zeta}_{m,l}\right) }.
	\end{align}

	%	

	%	 We introduce the following Theorem 1 to establish the relationship between the physical and spatial DoAs/ranges. 
	\begin{theorem}
		Denote $\mathbf{u}\in \mathbb{C}^{N_\mathrm{T}} $ and $\mathbf{v}_m \in \mathbb{C}^{N_\mathrm{T}}$ as the arbitrary near-field steering vectors corresponding to the physical (i.e., $\{\phi_{l},r_{l}\}$) and spatial (i.e., $\{\bar{\phi}_{m,l},	\bar{r}_{m,l}\}$) locations given in (\ref{steeringVectorPhy2}) and (\ref{steeringVectorSpa}), respectively. Then, in spatial domain at subcarrier frequency $f_m$, the array gain achieved by $\mathbf{u}^\textsf{H}\mathbf{v}_m$ is maximized and the generated beam is focused at the location $\{\bar{\phi}_{m,l},	\bar{r}_{m,l}\}$ such that  
		\begin{align}
		\label{physical_spatial_directions}
		\bar{\phi}_{m,l} =    \eta_m \phi_{l}, \hspace{5pt}
		\bar{r}_{m,l} =    \frac{1 - \eta_m^2 \phi_{l}^2}{\eta_m(1 -\phi_{l}^2)}r_{l},
		\end{align}
		where  	 $\eta_m = \frac{f_c}{f_m}$ represents the proportional deviation of DoA/ranges.
	\end{theorem}
	
	\begin{proof}
		Please see~\cite{elbir2023Feb_NF_THZ_CE}.
	\end{proof}

	\vspace{-5pt}
	
	Following (\ref{r_approx}) and (\ref{physical_spatial_directions}),  we define near-field beam-squint in terms of DoAs and ranges as, respectively,
	\begin{align}
	\label{beamSplit2}
	\Delta(\phi_{l},m) &= \bar{\phi}_{m,l} - \phi_{l} = (\eta_m -1)\phi_{l}, 
	\end{align}
	and $\Delta(r_{l},m) = \bar{r}_{m,l} - r_{l} = (\eta_m -1)r_{l}$, i.e., 
	\begin{align}
	\Delta(r_{l},m) =  (\eta_m -1) \frac{1 - \eta_m^2 \phi_{l}^2}{\eta_m(1 -\phi_{l}^2)}r_{l}.
	\end{align}
	%	In Fig.~\ref{fig_BS}, the array gain is computed for both far- and near-field cases, wherein the coordinates with maximum array gain are achieved at different locations for different subcarriers because of beam-split.

		\vspace{-15pt}
	\subsection{Radar Model}
	The aim of the radar sensing task is to achieve the highest SNR toward targets. Denote the estimate of the $k$-th  target direction and range by $\Phi_k$ and ${r}_k$, which can be estimated during the search phase of the radar, e.g., MUSIC algorithm~\cite{elbir2023Mar_ISAC_SPIM}. Then, we select the radar-only beamformer as
	\begin{align}
	\mathbf{F}_\mathrm{R} =  [\mathbf{a}_\mathrm{T}(\Phi_1,{r}_1),\cdots, \mathbf{a}_\mathrm{T}(\Phi_K,{r}_K)]\in \mathbb{C}^{N_\mathrm{T}\times K}.
	\end{align}
		The proposed ISAC beamformer aims to generate multiple beams toward both radar targets and the communication user. This allows us to maintain the communication between the user and the BS while tracking the radar targets, of which the initial directions/ranges are estimated.	Using the hybrid beamforming structure, the beampattern of the radar for $\Phi \in [-\frac{\pi}{2},\frac{\pi}{2}]$ and  $r\in [0, d_F]$ is 
	\begin{align}
	B_m(\Phi,r) = \mathrm{Trace}\{\mathbf{a}_\mathrm{T}^\textsf{H}(\Phi,r)\mathbf{R}_\mathbf{x}[m] \mathbf{a}_\mathrm{T}(\Phi, {r})  \},\label{beamPattern}
	\end{align}
	where $\mathbf{a}_\mathrm{T}(\Phi,r)\in\mathbb{C}^{N_\mathrm{T}}$ denotes the steering vector corresponding to arbitrary  direction $\Phi$ and range $r$, and  $\mathbf{R}_\mathbf{x}[m]\in \mathbb{C}^{N_\mathrm{T}\times N_\mathrm{T}}$ is the covariance of the transmit signal as 
	\begin{align}
	\mathbf{R}_\mathbf{x}[m] &= \mathbb{E}\{ \mathbf{x}[m]\mathbf{x}^{\textsf{H}}[m] \}
	\nonumber\\	
	& = \mathbf{F}_\mathrm{RF}\mathbf{F}_\mathrm{BB}[m]\mathbb{E}\{\mathbf{s}[m]\mathbf{s}^\textsf{H}[m]\} \mathbf{F}_\mathrm{BB}^{\textsf{H}}[m]\mathbf{F}_\mathrm{RF}^{\textsf{H}} 
	\nonumber \\	&
	= \frac{1}{N_\mathrm{S}}\mathbf{F}_\mathrm{RF}\mathbf{F}_\mathrm{BB}[m] \mathbf{F}_\mathrm{BB}^{\textsf{H}}[m]\mathbf{F}_\mathrm{RF}^{\textsf{H}}.
	\end{align}
	To simultaneously obtain the desired beampattern for the radar target and provide satisfactory communications performance, the hybrid beamformer $\mathbf{F}_\mathrm{RF}\mathbf{F}_\mathrm{BB}[m]$ should be designed accordingly.
	%, as discussed in the following.
	
	%	which requires designing hybrid beamformers. Next, we 

	\subsection{Problem Formulation}
	Our aim in this work is to design the ISAC hybrid beamformer $\mathbf{F}_\mathrm{RF}\mathbf{F}_\mathrm{BB}[m]$ while mitigating the impact of near-field beam-squint. The design problem maximizes the SE of the overall system, which can be recast via minimizing the Euclidean distance between the hybrid beamformer $\mathbf{F}_\mathrm{RF}\mathbf{F}_\mathrm{BB}[m]$ and the unconstrained JRC beamformer $\mathbf{F}_\mathrm{CR}[m]$~\cite{heath2016overview,elbir2021JointRadarComm}. The JRC beamformer is defined as 
	\begin{align}
	\label{Fcr}
	\mathbf{F}_\mathrm{CR}[m] = \varepsilon \mathbf{F}_\mathrm{opt}[m] + (1- \varepsilon) \mathbf{F}_\mathrm{R}\boldsymbol{\Pi}[m],
	\end{align} 
	where $\mathbf{F}_\mathrm{opt}[m]\in\mathbb{C}^{N_\mathrm{T}\times N_\mathrm{S}}$ is the unconstrained communications-only beamformer, which can be obtained through the singular value decomposition (SVD) of $\mathbf{H}[m]$~\cite{heath2016overview}.   $\boldsymbol{\Pi}[m]\in\mathbb{C}^{K\times N_\mathrm{S}}$ is a unitary matrix providing the change of dimensions between $\mathbf{F}_\mathrm{R}$ and $\mathbf{F}_\mathrm{opt}[m]$. In (\ref{Fcr}), $0\leq \varepsilon\leq 1$ represents  the trade-off parameter between the radar and communications tasks. In particular, $\varepsilon=1$ ($\varepsilon = 0$) corresponds to the communications-only (radar-only) design. In ISAC, $\varepsilon$ controls	the trade-off between the accuracy/prominence of sensing and communications tasks~\cite{elbir2022Aug_THz_ISAC}. 
	Now, the optimization problem becomes
	\begin{align}
	&\minimize_{\mathbf{F}_\mathrm{RF}, \mathbf{F}_\mathrm{BB}[m], \boldsymbol{\Pi} [m]} \; \sum_{m = 1}^{M} \| \mathbf{F}_\mathrm{RF}\mathbf{F}_\mathrm{BB}[m] -  \mathbf{F}_\mathrm{CR}[m] \|_\mathcal{F} \nonumber \\
	&\hspace{20pt} \subjectto  | [\mathbf{F}_\mathrm{RF}]_{i,j} | =1/\sqrt{N_\mathrm{T}} \nonumber\\
	&\hspace{20pt} \sum_{m = 1}^{M}\| \mathbf{F}_\mathrm{RF}{\mathbf{F}}_\mathrm{BB}[m] \|_\mathcal{F} = MN_\mathrm{S} \nonumber \\
	& \hspace{20pt} \boldsymbol{\Pi}[m]{ \boldsymbol{\Pi}}^\textsf{H}[m] = \mathbf{I}_{K}. \label{problemOpt}
	\end{align}
	The above optimization problem is difficult to solve due to non-convex constraints, e.g., unit-modulus constraint and it involves multiple unknowns   $\mathbf{F}_\mathrm{RF}$, $\mathbf{F}_\mathrm{BB}[m]$ and $\boldsymbol{\Pi}[m]$. In order to provide an effective solution, we follow an alternating optimization approach, wherein the beamformers are optimized one-by-one while the other term is fixed. Specifically, $\mathbf{F}_\mathrm{RF}$ is first estimated via an OMP-based approach, wherein the columns of $\mathbf{F}_\mathrm{RF}$ are selected from a dictionary of near-field steering vectors. Next, the baseband beamformer   $\mathbf{F}_\mathrm{BB}[m]$ and $\boldsymbol{\Pi}[m]$ are estimated. Finally, a BSA baseband beamformer is designed for beam-squint compensation.

	%	Various methods have been proposed to solve (\ref{opt1})~\cite{heath2016overview} such as OMP~\cite{limitedFeedback_Alkhateeb2015Jul} and MO~\cite{hybridBFAltMin}. However, all of these algorithms fail to take into account the impact of beam-split due to the usage of SI AB. In what follows, we introduce our BSA hybrid beamforming approach to efficiently mitigate beam-split by using SI ABs without any additional hardware components.

	%	\vspace{-12pt}
	\section{Hybrid Beamformer Design}
	\label{sec:SPIMISAC}
	In order to solve (\ref{problemOpt}) effectively, we propose an alternating algorithm to efficiently find the unknowns $\mathbf{F}_\mathrm{RF}, \mathbf{F}_\mathrm{BB}[m], \boldsymbol{\Pi} [m]$. Thus, we first introduce an OMP based approach, wherein the analog beamformer $\mathbf{F}_\mathrm{RF}$ is designed, respectively, from the columns of the  dictionary matrix
	\begin{align}
	\mathbf{D}= [\mathbf{a}_\mathrm{T}(\phi_1,r_1),\cdots, \mathbf{a}_\mathrm{T}(\phi_N,r_N)]\in \mathbb{C}^{N_\mathrm{T}\times N},
	\end{align}
	where $N$ is the grid size of the dictionary with $\phi_n \in [-1,1]$, $r_n\in (0, d_F]$.
	Then, the columns of the analog beamformer $\mathbf{F}_\mathrm{RF}$ are selected from the columns of $\mathbf{{\mathbf{D}}}[m]$ as $\mathbf{a}_\mathrm{T}(\phi_{p^*},r_{p^*})$, for $\ell = 1,\cdots, N_\mathrm{RF}$  where
	\begin{align}
	p^\star = \argmax_{p \in \{1,\cdots, N\}} \sum_{m=1}^{M}\left|\left[\boldsymbol{\Psi}[m]\boldsymbol{\Psi}^\textsf{H}[m]\right]_{p,p} \right|,
	\end{align}
	where  	$\boldsymbol{\Psi}[m] = \mathbf{a}_\mathrm{T}^\textsf{H}(\phi_p,r_p) \mathbf{F}_\mathrm{CR}[m]$.
	Once  the analog beamformer $\mathbf{F}_\mathrm{RF}$ is obtained and by using  $\mathbf{F}_\mathrm{CR}[m]$, the baseband beamformer  is given by
	\begin{align}
	\mathbf{F}_\mathrm{BB}[m] = {\mathbf{F}_\mathrm{RF}}^\dagger \mathbf{F}_\mathrm{CR}[m],
	\end{align}
	which is then normalized as $\mathbf{F}_\mathrm{BB}[m] = \frac{\sqrt{N_\mathrm{S}} \mathbf{F}_\mathrm{RF}^\dagger \mathbf{F}_\mathrm{CR}[m]  }{\|\mathbf{F}_\mathrm{RF}\mathbf{F}_\mathrm{BB}[m]    \|_\mathcal{F}}$. The JRC beamformer is composed of the auxiliary matrix $\boldsymbol{\Pi}[m]$, which can be optimized as 
	\begin{align}
	\label{prob_FBB_P}
	&\minimize_{\overline{\boldsymbol{\Pi}}} \hspace{3pt} \|\mathbf{F}_\mathrm{RF}\overline{\mathbf{F}}_\mathrm{BB} - {\mathbf{F}}_\mathrm{CR}   \|_\mathcal{F}^2  \nonumber \\
	&	\subjectto \hspace{5pt} \overline{\boldsymbol{\Pi}}\; \overline{\boldsymbol{\Pi}}^\textsf{H} = \mathbf{I}_K,
	\end{align}
	where $\overline{\mathbf{F}}_\mathrm{BB} = \left[ \mathbf{F}_\mathrm{BB}[1], \cdots, \mathbf{F}_\mathrm{BB}[M] \right]$, $\overline{\mathbf{F}}_\mathrm{CR} = \left[ \mathbf{F}_\mathrm{CR}[1], \cdots, \mathbf{F}_\mathrm{CR}[M] \right]$ and $\overline{\boldsymbol{\Pi}} = \left[\boldsymbol{\Pi}[1],\cdots, \boldsymbol{\Pi}[M] \right]$ are $N_\mathrm{RF}\times MN_\mathrm{S}$, $N_\mathrm{T}\times MN_\mathrm{S}$ and $K\times MN_\mathrm{S}$ matrices composed of information corresponding to all subcarriers, respectively. The solution to the problem in (\ref{prob_FBB_P}) can be found via SVD of the $K\times MN_\mathrm{S}$ matrix $\mathbf{F}_\mathrm{R}^\textsf{H} \mathbf{F}_\mathrm{RF} \overline{\mathbf{F}}_\mathrm{BB}$ and it is given by
	\begin{align}
	\label{pi_m}
	\overline{\boldsymbol{\Pi}} = \widetilde{\boldsymbol{\Pi}} \mathbf{I}_{K\times MN_\mathrm{S}} \widetilde{\mathbf{V}},
	\end{align}
	where $\widetilde{\boldsymbol{\Pi}} \widetilde{\boldsymbol{\Sigma}} \widetilde{\mathbf{V}}  =  \mathbf{F}_\mathrm{R}^\textsf{H} \mathbf{F}_\mathrm{RF} \overline{\mathbf{F}}_\mathrm{BB}$ is the SVD of the $N_\mathrm{RF}\times N_\mathrm{S}$ matrix $\frac{1}{1- \varepsilon }\mathbf{F}_\mathrm{R}^\textsf{H} \left(\mathbf{F}_\mathrm{RF} \overline{\mathbf{F}}_\mathrm{BB} -  \varepsilon \overline{\mathbf{F}}_\mathrm{CR}\right)$, and $\mathbf{I}_{K\times MN_\mathrm{S}}  = \left[\mathbf{I}_K \hspace{1pt}|  \hspace{1pt} \mathbf{0}_{ MN_\mathrm{S}- K\times K}^\textsf{T}  \right]^\textsf{T}$. Then, by estimating $\mathbf{F}_\mathrm{BB}[m]$ and $\boldsymbol{\Pi}[m]$ iteratively, the hybrid beamformer weights are computed.

	%-------------------------------------------------------------------------------------------------
	\begin{algorithm}[t]
		\begin{algorithmic}[1] 
			\caption{ \bf ISAC hybrid beamforming}
			\Statex {\textbf{Input:} $\mathbf{D}$, $\mathbf{F}_\mathrm{R}$, $\mathbf{F}_\mathrm{opt}[m]$, $\varepsilon$, $\eta_m$. \label{alg:BSAHB}}
			\State  $\mathbf{F}_\mathrm{RF} = \mathrm{Empty}$, $\mathbf{F}_\mathrm{res}[m] =\mathbf{F}_\mathrm{CR}[m]$.
			%			\State  $\tilde{\mathbf{D}}[m] = 1/\sqrt{N_\mathrm{T}}\exp\{\mathrm{j}\boldsymbol{\Xi} \eta_m\}  $.
			\State \textbf{for} $\ell =1,\cdots, N_\mathrm{RF}$ \textbf{do}
			\State \indent  $p^\star = \argmax_{p} \sum_{m=1}^{M}\left|\mathbf{a}_\mathrm{T}^\textsf{H}(\phi_p,r_p) \mathbf{F}_\mathrm{res}[m] \right|$.
			%		, where \par  \indent  $\mathbf{d}_{p}[m] = [\tilde{\mathbf{D}}[m]]_p$.
			\State \indent $\mathbf{F}_\mathrm{RF} =\left[\mathbf{F}_\mathrm{RF}| \mathbf{a}_\mathrm{T}(\phi_{p^*},r_{p^*}) \right]$.
			\State  \indent$\mathbf{F}_\mathrm{BB}[m] = (\mathbf{F}_\mathrm{RF}^\textsf{H}\mathbf{F}_\mathrm{RF})^{-1}\mathbf{F}_\mathrm{RF}^\textsf{H} \mathbf{F}_\mathrm{CR}[m].$
			\State \indent Update $\boldsymbol{\Pi}[m]$ from (\ref{pi_m}).
			\State \indent Update $\mathbf{F}_\mathrm{CR}[m]$ from (\ref{Fcr}).
			\State \indent $\mathbf{F}_\mathrm{res}[m] = 		\frac{ \mathbf{F}_\mathrm{CR}[m] - \mathbf{F}_\mathrm{RF}\mathbf{F}_\mathrm{BB}[m] }{ \| \mathbf{F}_\mathrm{CR}[m] - \mathbf{F}_\mathrm{RF}\mathbf{F}_\mathrm{BB}[m]  \|_\mathcal{F} }.$
			\State \textbf{end for}
			\State $\mathbf{F}_\mathrm{BB}[m] = 		\sqrt{N_\mathrm{S}}\frac{ \mathbf{F}_\mathrm{BB}[m] }{ \| \mathbf{F}_\mathrm{RF}\mathbf{F}_\mathrm{BB}[m]  \|_\mathcal{F} }.$
			\State  $\breve{\mathbf{F}}_\mathrm{RF}[m] = \frac{1}{\sqrt{N_\mathrm{T}}}   \boldsymbol{\Omega}[m]$ where $[\boldsymbol{\Omega}[m]]_{i,j} = \exp \{\mathrm{j} {\eta_m} \angle \{[\mathbf{F}_\mathrm{RF}]_{i,j} \}\}$.
			\State   $\widetilde{\mathbf{F}}_\mathrm{BB}[m] = \mathbf{F}_\mathrm{RF}^\dagger \breve{\mathbf{F}}_\mathrm{RF}[m] \mathbf{F}_\mathrm{BB}[m]$.
			% \State  $\widetilde{\mathbf{F}}[m] = 		\sqrt{N_\mathrm{S}}\frac{ \widetilde{\mathbf{F}}_\mathrm{BB}[m] }{ \| \mathbf{F}_\mathrm{RF}\widetilde{\mathbf{F}}_\mathrm{BB}[m]  \|_\mathcal{F} }.$
			
			\Statex \textbf{Return:} $\mathbf{F}_\mathrm{RF}$, $\widetilde{\mathbf{F}}_\mathrm{BB}[m]$.
		\end{algorithmic} 
	\end{algorithm}
	%------------------------------------------------------------------------------------------------
	
	The next task is to mitigate  near-field beam-squint, which can be compensated if SD analog beamformers are used. However, this approach is costly since it requires employing $MN_\mathrm{T}N_\mathrm{RF}$ (instead of $N_\mathrm{T}N_\mathrm{RF}$) phase-shifters. Instead, we propose an efficient approach, wherein the effect of beam-squint is handled in the baseband beamformer, which is SD. Therefore, the effect of beam-squint is conveyed from analog domain to baseband.
	
	Denoted by $\breve{\mathbf{F}}_\mathrm{RF}[m]\in \mathbb{C}^{N_\mathrm{T}\times N_\mathrm{RF}}$,  the SD analog beamformer that can be computed from the SI analog beamformer $\mathbf{F}_\mathrm{RF}$ as 
	\begin{align}
	\breve{\mathbf{F}}_\mathrm{RF}[m] = \frac{1}{\sqrt{N_\mathrm{T}}}   \boldsymbol{\Omega}[m],
	\end{align}
	where $\boldsymbol{\Omega}[m]\in \mathbb{C}^{N_\mathrm{T}\times N_\mathrm{RF}}$ includes the angle information of $\mathbf{F}_\mathrm{RF}$ as $[\boldsymbol{\Omega}[m]]_{i,j} = \exp \{\mathrm{j} \eta_m \angle \{[\mathbf{F}_\mathrm{RF}]_{i,j} \}\}$ for $i = 1,\cdots, N_\mathrm{T}$ and $j =1,\cdots, N_\mathrm{RF}$. As a result, the angular deviation in $\mathbf{F}_\mathrm{RF}$ due to beam-squint is compensated  with $\eta_m$. 
	
	Now, we define $\widetilde{\mathbf{F}}_\mathrm{BB}[m]\in \mathbb{C}^{N_\mathrm{RF}\times N_\mathrm{S}}$ as the \textit{BSA digital beamformer} in order to achieve SD beamforming performance that can be obtained by the usage of SD analog beamformer $\breve{\mathbf{F}}_\mathrm{RF}[m]$. Hence, we aim to match the proposed \textit{BSA hybrid beamformer} $\mathbf{F}_\mathrm{RF} \widetilde{\mathbf{F}}_\mathrm{BB}[m]$ with the SD hybrid beamformer $\breve{\mathbf{F}}_\mathrm{RF}[m] \mathbf{F}_\mathrm{BB}[m] $ as
	\begin{align}
	\minimize_{\widetilde{\mathbf{F}}_\mathrm{BB}[m]} \| \mathbf{F}_\mathrm{RF} \widetilde{\mathbf{F}}_\mathrm{BB}[m] - \breve{\mathbf{F}}_\mathrm{RF}[m] \mathbf{F}_\mathrm{BB}[m] \|_\mathcal{F},
	\end{align}
	for which $\widetilde{\mathbf{F}}_\mathrm{BB}[m]$ can be obtained as
	\begin{align}
	\label{fbbTilde}
	\widetilde{\mathbf{F}}_\mathrm{BB}[m] = {\mathbf{F}_\mathrm{RF}}^\dagger \breve{\mathbf{F}}_\mathrm{RF}[m] \mathbf{F}_\mathrm{BB}[m].
	\end{align}
	Because of the reduced dimension of the baseband beamformer (i.e., $N_\mathrm{RF}< N_\mathrm{T}$), the BSA approach does not completely mitigate beam-squint. In other words,  the beam-squint can be fully mitigated only if  ${\mathbf{F}_\mathrm{RF}}^\dagger \breve{\mathbf{F}}_\mathrm{RF}[m]  = \mathbf{I}_{\mathrm{N}_\mathrm{T}}$ so that the resulting hybrid beamformer $\mathbf{F}_\mathrm{RF} \widetilde{\mathbf{F}}_\mathrm{BB}[m]$ can be equal to  $\breve{\mathbf{F}}_\mathrm{RF}[m] \mathbf{F}_\mathrm{BB}[m]$,  which requires $N_\mathrm{RF} = N_\mathrm{T}$. Nevertheless, the proposed approach provides satisfactory SE performance with beam-squint compensation for a wide range of bandwidth~\cite{elbir2023Mar_ISAC_SPIM,elbir2022_thz_beamforming_Unified_Elbir2022Sep}. Finally, the algorithmic steps of the proposed hybrid beamforming approach are  presented in Algorithm~\ref{alg:BSAHB}, wherein we select the columns of the analog beamformer $\mathbf{F}_\mathrm{RF}$ from the near-field dictionary $\mathbf{D}$ for $\ell = 1,\cdots, N_\mathrm{RF}$. In this process, the similarity between the columns of $\mathbf{D}$ (i.e., $\mathbf{a}_\mathrm{T}(\phi_p,r_p)$) and the residual beamformer (i.e., $\mathbf{F}_\mathrm{res}[m]$) is performed. Since this process is iterated for $\ell = 1,\cdots, N_\mathrm{RF}$, its convergence is similar to the previous works~\cite{heath2016overview,mimoRHeath,mimoHybridLeus3}.
	
	%	\vspace{-12pt}
	\section{Numerical Experiments}
	\label{sec:Sim}
	We evaluated the performance of our hybrid beamforming technique in comparison with the fully digital (DF) ISAC and communications-only beamformers as well as far-field-based design, in terms of SE, averaged over $500$ Monte Carlo trials. The number of antennas at the BS and the user are $N_\mathrm{T}=128$ and $N_\mathrm{R}=16$, respectively. The  carrier frequency and the bandwidth are selected as $f_c = 300$ GHz and $B = 20$ GHz, respectively,
	%	 $60$ GHz and $3$ GHz ($300$ GHz and $30$ GHz) for mmWave (THz) scenario
	and the number of subcarriers is $M=64$. The number of targets is $K=3$, number of spatial paths is $L=8$, number of RF chains is $N_\mathrm{RF} = 8$ and the trade-off parameter is $\varepsilon = 0.5$. The dictionary grid size is obtained from $N_\phi = 100$, $N_r = 20$, yielding $N = N_\phi N_r = 2000$. Targets and path directions (ranges) are uniformly drawn at random from the intervals $[-\frac{\pi}{3},\frac{\pi}{3}]$ ($[5, 30]$ m)~\cite{elbir2021JointRadarComm}. 
	
	%%-----------------------------------------------------
	\begin{figure}
		\centering
		{\includegraphics[draft=false,width=\columnwidth]{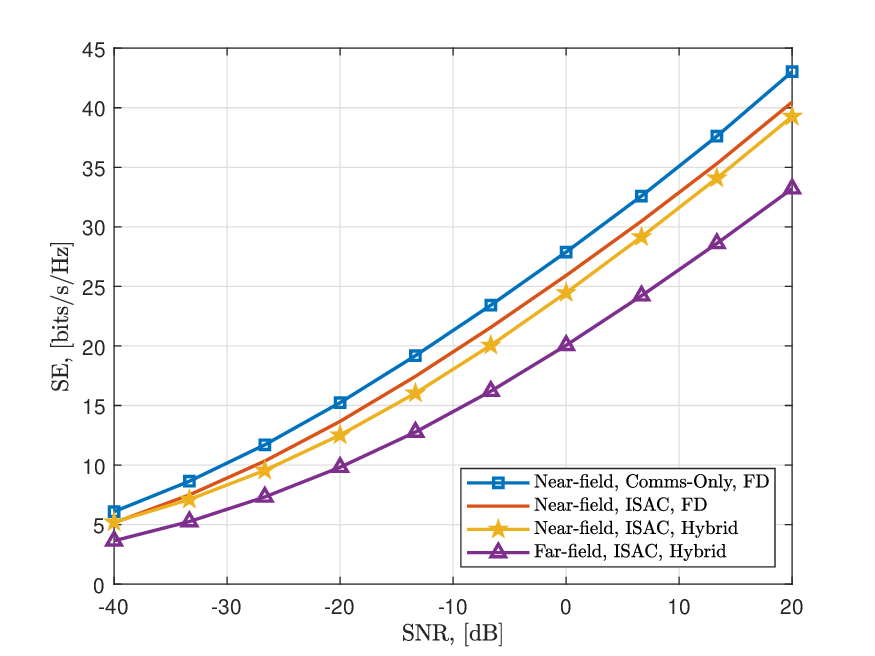} } 
		% 	\subfloat[]{\includegraphics[draft=false,width=.33\textwidth]{BS_frequency.eps} } 
		\caption{SE performance versus SNR.   	}
		%			\vspace*{-5mm}
		\label{fig_SE_SNR}
	\end{figure}
	%%-----------------------------------------------------
	
	Fig.~\ref{fig_SE_SNR} delineates the SE performance of the competing algorithms. We can see that the communications-only ($\varepsilon = 1$) FD beamformer provides the highest SE while the ISAC ($\varepsilon = 0.5$) FD beamformer provides marginally reduced SE  due to power allocation for both communications and sensing tasks. The proposed beamforming approach exhibits performance closely resembling that of the FD beamformers. A significant performance degradation is also observed when the near-field  model is overlooked in favor of the far-field array model.
	
	%%-----------------------------------------------------
	\begin{figure}
		\centering
		{\includegraphics[draft=false,width=\columnwidth]{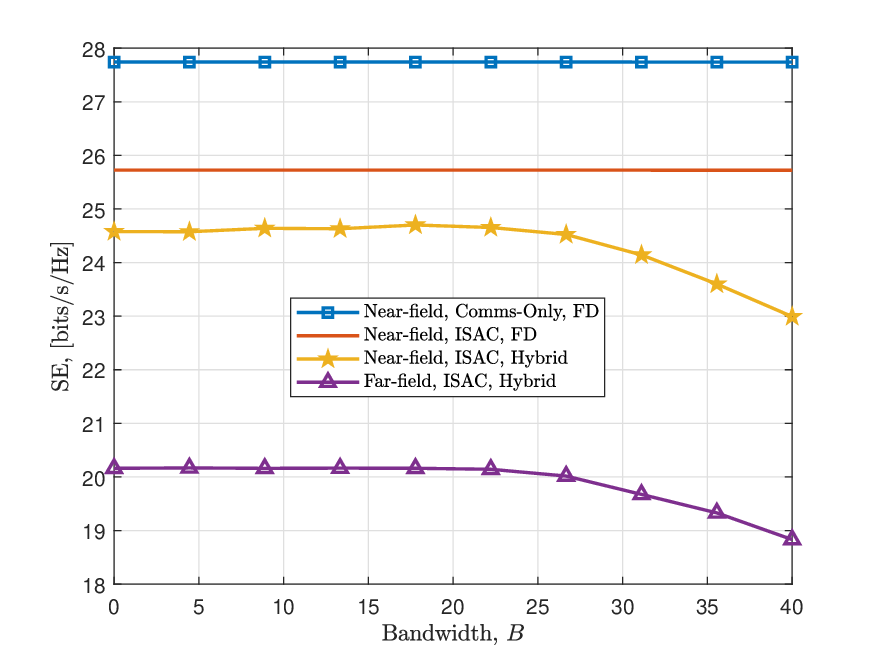} } 
		% 	\subfloat[]{\includegraphics[draft=false,width=.33\textwidth]{BS_frequency.eps} } 
		\caption{SE performance versus bandwidth.   	}
		%			\vspace*{-5mm}
		\label{fig_SE_BW}
	\end{figure}
	%%-----------------------------------------------------
	
	Fig.~\ref{fig_SE_BW} shows the SE performance against the bandwidth $B\in [0,40]$ GHz. We can see that the proposed hybrid beamforming scheme achieves satisfactory SE performance up to $B \leq 30$ GHz, beyond which its performance slightly declines. This degradation arises because the BSA baseband beamformer's low-dimensional structure cannot adequately compensate for the beam-squint. Nevertheless, %it closely follows the performance of the FD beamforming and with  much higher SE improvement as compared to far-field model.
	the hybrid beamforming scheme closely trails the performance of the FD beamforming and yields a substantial SE improvement compared to the far-field model.
	
	\section{Summary}
	We introduced a hybrid beamforming scheme for THz-ISAC systems in near-field scenario. The analog beamformers are designed based on a dictionary composed of near-field steering vectors. Then, the baseband and JRC beamformers are obtained. In order to cope with beam-squint problem in near-field scenario, we utilized the baseband beamformer to convey the impact of beam-squint from analog domain to baseband without requiring additional hardware components. As future work, we reserve to study a more challenging scenario, e.g., near-field ISAC joint precoder and combiner design.

	%\balance 
	\footnotesize
	\bibliographystyle{IEEEtran}
	\bibliography{IEEEabrv,references_121,references_119}

\begin{thebibliography}{10}
\providecommand{\url}[1]{#1}
\csname url@samestyle\endcsname
\providecommand{\newblock}{\relax}
\providecommand{\bibinfo}[2]{#2}
\providecommand{\BIBentrySTDinterwordspacing}{\spaceskip=0pt\relax}
\providecommand{\BIBentryALTinterwordstretchfactor}{4}
\providecommand{\BIBentryALTinterwordspacing}{\spaceskip=\fontdimen2\font plus
\BIBentryALTinterwordstretchfactor\fontdimen3\font minus
  \fontdimen4\font\relax}
\providecommand{\BIBforeignlanguage}[2]{{%
\expandafter\ifx\csname l@#1\endcsname\relax
\typeout{** WARNING: IEEEtran.bst: No hyphenation pattern has been}%
\typeout{** loaded for the language `#1'. Using the pattern for}%
\typeout{** the default language instead.}%
\else
\language=\csname l@#1\endcsname
\fi
#2}}
\providecommand{\BIBdecl}{\relax}
\BIBdecl

\bibitem{mishra2019toward}
K.~V. Mishra, M.~R.~B. Shankar, V.~Koivunen, B.~Ottersten, and S.~A. Vorobyov,
  ``Toward millimeter-wave joint radar communications: {A} signal processing
  perspective,'' \emph{IEEE Signal Process. Mag.}, vol.~36, no.~5, pp.
  100--114, 2019.

\bibitem{elbir2022Aug_THz_ISAC}
A.~M. Elbir, K.~V. Mishra, S.~Chatzinotas, and M.~Bennis, ``{Terahertz-band
  integrated sensing and communications: Challenges and opportunities},''
  \emph{arXiv preprint arXiv:2208.01235}, Aug. 2022.

\bibitem{elbir2021JointRadarComm}
A.~M. Elbir, K.~V. Mishra, and S.~Chatzinotas, ``Terahertz-band joint
  ultra-massive {MIMO} radar-communications: {M}odel-based and model-free
  hybrid beamforming,'' \emph{IEEE J. Sel. Top. Signal Process.}, vol.~15,
  no.~6, pp. 1468--1483, 2021.

\bibitem{heath2016overview}
R.~W. Heath, N.~Gonz{\ifmmode\acute{a}\else\'{a}\fi}lez-Prelcic, S.~Rangan,
  W.~Roh, and A.~M. Sayeed, ``An overview of signal processing techniques for
  millimeter wave {MIMO} systems,'' \emph{IEEE J. Sel. Top. Signal Process.},
  vol.~10, no.~3, pp. 436--453, 2016.

\bibitem{elbir2022Nov_Beamforming_SPM}
A.~M. Elbir, K.~V. Mishra, S.~A. Vorobyov, and R.~W. Heath, ``{Twenty-Five
  Years of Advances in Beamforming: From convex and nonconvex optimization to
  learning techniques},'' \emph{IEEE Signal Process. Mag.}, vol.~40, no.~4, pp.
  118--131, Jun. 2023.

\bibitem{alkhateeb2016frequencySelective}
A.~{Alkhateeb} and R.~W. {Heath}, ``Frequency selective hybrid precoding for
  limited feedback millimeter wave systems,'' \emph{IEEE Trans. Commun.},
  vol.~64, no.~5, pp. 1801--1818, 2016.

\bibitem{elbir_THZ_CE_ArrayPerturbation_Elbir2022Aug}
A.~M. Elbir, W.~Shi, A.~K. Papazafeiropoulos, P.~Kourtessis, and
  S.~Chatzinotas, ``{Terahertz-Band Channel and Beam Split Estimation via Array
  Perturbation Model},'' \emph{IEEE Open J. Commun. Soc.}, vol.~4, pp.
  892--907, Mar. 2023.

\bibitem{beamSquint_FeiFei_Wang2019Oct}
B.~Wang, M.~Jian, F.~Gao, G.~Y. Li, and H.~Lin, ``Beam squint and channel
  estimation for wideband {mmWave} massive {MIMO-OFDM} systems,'' \emph{IEEE
  Trans. Signal Process.}, vol.~67, no.~23, pp. 5893--5908, 2019.

\bibitem{elbir_BSA_OMP_THZ_CE_Elbir2023Feb}
A.~M. Elbir and S.~Chatzinotas, ``{BSA-OMP: Beam-Split-Aware Orthogonal
  Matching Pursuit for THz Channel Estimation},'' \emph{IEEE Wireless Commun.
  Lett.}, vol.~12, no.~4, pp. 738--742, Feb. 2023.

\bibitem{trueTimeDelayBeamSquint}
F.~Gao, B.~Wang, C.~Xing, J.~An, and G.~Y. Li, ``Wideband beamforming for
  hybrid massive {MIMO} terahertz communications,'' \emph{IEEE J. Sel. Areas
  Commun.}, vol.~39, no.~6, pp. 1725--1740, 2021.

\bibitem{beamSquintAwareHB_SD_You2022Aug}
L.~You, X.~Qiang, C.~G. Tsinos, F.~Liu, W.~Wang, X.~Gao, and B.~Ottersten,
  ``Beam squint-aware integrated sensing and communications for hybrid massive
  {MIMO} {LEO} satellite systems,'' \emph{IEEE J. Sel. Areas Commun.}, vol.~40,
  no.~10, pp. 2994--3009, 2022.

\bibitem{elbir2023Mar_ISAC_SPIM}
A.~M. Elbir, K.~V. Mishra, A.~Celik, and A.~M. Eltawil, ``{Millimeter-Wave
  Radar Beamforming with Spatial Path Index Modulation Communications},'' in
  \emph{{2023 IEEE Radar Conference (RadarConf23)}}.\hskip 1em plus 0.5em minus
  0.4em\relax IEEE, May 2023, pp. 1--6.

\bibitem{nf_primer_Bjornson2021Oct}
E.~Bj{\ifmmode\ddot{o}\else\"{o}\fi}rnson, {\ifmmode\ddot{O}\else\"{O}\fi}.~T.
  Demir, and L.~Sanguinetti, ``{A Primer on Near-Field Beamforming for Arrays
  and Reconfigurable Intelligent Surfaces},'' in \emph{{2021 55th Asilomar
  Conference on Signals, Systems, and Computers}}.\hskip 1em plus 0.5em minus
  0.4em\relax IEEE, Oct. 2021, pp. 105--112.

\bibitem{nf_OMP_Dai_Wei2021Nov}
X.~Wei and L.~Dai, ``{Channel Estimation for Extremely Large-Scale Massive
  MIMO: Far-Field, Near-Field, or Hybrid-Field?}'' \emph{IEEE Commun. Lett.},
  vol.~26, no.~1, pp. 177--181, Nov. 2021.

\bibitem{nf_mmwave_CE_noBeamSplit_Cui2022Jan}
M.~Cui and L.~Dai, ``{Channel Estimation for Extremely Large-Scale MIMO:
  Far-Field or Near-Field?}'' \emph{IEEE Trans. Commun.}, vol.~70, no.~4, pp.
  2663--2677, Jan. 2022.

\bibitem{nf_NB2_Zhang2022Nov}
X.~Zhang, Z.~Wang, H.~Zhang, and L.~Yang, ``{Near-Field Channel Estimation for
  Extremely Large-Scale Array Communications: A model-based deep learning
  approach},'' \emph{arXiv preprint arXiv:2211.15440}, Nov. 2022.

\bibitem{elbir2022Jul_THz_CE_FL}
A.~M. Elbir, W.~Shi, K.~V. Mishra, and S.~Chatzinotas, ``{Federated Multi-Task
  Learning for THz Wideband Channel and DoA Estimation},'' \emph{arXiv preprint
  arXiv:2207.06017}, Jul. 2022.

\bibitem{dovelos_THz_CE_channelEstThz2}
K.~Dovelos, M.~Matthaiou, H.~Q. Ngo, and B.~Bellalta, ``Channel estimation and
  hybrid combining for wideband terahertz massive {MIMO} systems,'' \emph{IEEE
  J. Sel. Areas Commun.}, vol.~39, no.~6, pp. 1604--1620, 2021.

\bibitem{elbir2022_thz_beamforming_Unified_Elbir2022Sep}
A.~M. Elbir, ``{A Unified Approach for Beam-Split Mitigation in Terahertz
  Wideband Hybrid Beamforming},'' \emph{IEEE Trans. Veh. Technol.}, pp. 1--6,
  Apr. 2023.

\bibitem{elbir2023Feb_NF_THZ_CE}
A.~M. Elbir, W.~Shi, A.~K. Papazafeiropoulos, P.~Kourtessis, and
  S.~Chatzinotas, ``{Near-Field Terahertz Communications: Model-Based and
  Model-Free Channel Estimation},'' \emph{IEEE Access}, vol.~11, pp.
  36\,409--36\,420, Apr. 2023.

\bibitem{elbir2023Feb_NF_THZ_CE_ICASSP_NBAOMP}
A.~M. Elbir, K.~V. Mishra, and S.~Chatzinotas, ``{NBA-OMP: Near-Field
  Beam-Split-Aware Orthogonal Matching Pursuit for Wideband THz Channel
  Estimation},'' in \emph{{ICASSP 2023 - 2023 IEEE International Conference on
  Acoustics, Speech and Signal Processing (ICASSP)}}.\hskip 1em plus 0.5em
  minus 0.4em\relax IEEE, Jun. 2023, pp. 1--5.

\bibitem{nf_ISAC_1Wang2023Feb}
Z.~Wang, X.~Mu, and Y.~Liu, ``{Near-Field Integrated Sensing and
  Communications},'' \emph{arXiv}, Feb. 2023.

\bibitem{ummimoTareqOverview}
H.~Sarieddeen, M.-S. Alouini, and T.~Y. Al-Naffouri, ``An overview of signal
  processing techniques for terahertz communications,'' \emph{Proc. IEEE}, vol.
  109, no.~10, pp. 1628--1665, 2021.

\bibitem{ummimoHBThzSVModel}
H.~Yuan, N.~Yang, K.~Yang, C.~Han, and J.~An, ``Hybrid beamforming for
  terahertz multi-carrier systems over frequency selective fading,'' \emph{IEEE
  Trans. Commun.}, vol.~68, no.~10, pp. 6186--6199, 2020.

\bibitem{thz_mmWave_path_Comparison_Yan2020Jun}
L.~Yan, C.~Han, and J.~Yuan, ``{A Dynamic Array-of-Subarrays Architecture and
  Hybrid Precoding Algorithms for Terahertz Wireless Communications},''
  \emph{IEEE J. Sel. Areas Commun.}, vol.~38, no.~9, pp. 2041--2056, 2020.

\bibitem{nf_Fresnel_Cui2022Nov}
M.~Cui, L.~Dai, Z.~Wang, S.~Zhou, and N.~Ge, ``{Near-Field Rainbow: Wideband
  Beam Training for XL-MIMO},'' \emph{IEEE Trans. Wireless Commun.}, p.~1, Nov.
  2022.

\bibitem{mimoRHeath}
O.~E. Ayach, S.~Rajagopal, S.~Abu-Surra, Z.~Pi, and R.~W. Heath, ``Spatially
  sparse precoding in millimeter wave {MIMO} systems,'' \emph{IEEE Trans.
  Wireless Commun.}, vol.~13, no.~3, pp. 1499--1513, 2014.

\bibitem{mimoHybridLeus3}
A.~Alkhateeb, G.~Leus, and R.~W. Heath, ``Limited feedback hybrid precoding for
  multi-user millimeter wave systems,'' \emph{{IEEE} Trans. Wireless Commun.},
  vol.~14, no.~11, pp. 6481--6494, 2015.

\end{thebibliography}

\end{document}